\newif\ifllncs\llncsfalse
\newif\ifanon\anonfalse

\ifllncs
\documentclass{llncs}
\else
\documentclass[11pt,letter]{article}
\fi 

\usepackage{breakcites}
\usepackage[utf8]{inputenc}
\usepackage[T1]{fontenc}
\ifllncs
\else
\usepackage{times}
\usepackage{fullpage}
\fi 

\usepackage{mdframed}
\usepackage{amsmath}
\usepackage{relsize}
\usepackage{amsfonts}
\usepackage{amssymb}
\usepackage[dvipsnames]{xcolor}
\usepackage{graphicx}
\usepackage{braket}
\usepackage{url}
\usepackage{bm}
\usepackage{tikz}
\usetikzlibrary{shapes,arrows,positioning,calc}
\usepackage[mathscr]{euscript}
\allowdisplaybreaks

\usepackage{soul}
\usepackage{multirow}
\definecolor{DarkBlue}{RGB}{0,0,150}
\definecolor{NotSoDarkBlue}{RGB}{15,15,210}
\definecolor{DarkRed}{RGB}{150,0,0}
\definecolor{DarkGreen}{RGB}{0,100,0}
\usepackage[pdfstartview=FitH,pdfpagemode=UseNone,colorlinks,linkcolor=DarkRed,filecolor=blue,citecolor=DarkRed,urlcolor=DarkRed,pagebackref,breaklinks]{hyperref}
\usepackage[ruled, algosection]{algorithm2e}
\usepackage{cleveref}

\usepackage{orcidlink}

\usetikzlibrary{quantikz}

\newcommand{\CC}{\mathbb{C}}

\newcommand{\one}{\mathbf{1}}
\newcommand{\poly}{\mathsf{poly}}

\newcommand{\norm}[1]{\left\| {#1} \right\|}

\renewcommand{\vec}[1]{\mathbf{#1}}

\usepackage{amsthm}
\newtheorem{theorem}{Theorem}
\newtheorem{maintheorem}{Main Theorem}

\newtheorem{lemma}[theorem]{Lemma}
\newtheorem{corollary}[theorem]{Corollary}
\newtheorem{definition}[theorem]{Definition}
\newtheorem{remark}{Remark}

\newtheorem*{theorem*}{Theorem}
\numberwithin{theorem}{section}
\numberwithin{conjecture}{section}
\numberwithin{problem}{section}
 \newtheorem{construction}{Construction}

\newmdtheoremenv[backgroundcolor=gray!10,
                 linewidth=0pt,
                 innerleftmargin=16pt,
                 innerrightmargin=16pt,
                 innertopmargin=6pt,
                 innerbottommargin=6pt,
            splitbottomskip=4pt]{protocol}[prot]{Game}

\newcommand{\bit}{\{0,1\}}
\newcommand\algo{\mathcal}

\newcommand{\reg}[1]{{\color{gray}\mathsf{#1}}}
\newcommand{\linear}{\mathrm{L}}

\newcommand{\negl}{\mathsf{negl}}
\newcommand{\ketbra}[2]{\left|#1\right\rangle\!\!\left\langle #2\right|}
\newcommand{\Tr}[1]{\mathrm{Tr}\left[#1 \right]}  
\newcommand{\ot}{\otimes}
\newcommand\id{\mathbb{I}}

\newif\ifnotes
\notesfalse

\title{Efficient Unclonable Encryption from Pauli Eigenstates}
\ifanon
    \author{Anonymous Submission}
\else
    \author{Seyoon Ragavan\\MIT \\ \url{sragavan@mit.edu}}
    \date{\today}
\fi

\newcommand{\F}{\mathbb{F}}

\begin{document}
\maketitle

\begin{abstract}
We give, to our knowledge, the first plain-model, one-time information-theoretically secure, efficient unclonable encryption scheme for one classical bit. Previous work by Bhattacharyya and Culf (Nature Physics, 2026) and Bhattacharyya, Broadbent, and Culf (arXiv:2603.08916) either only showed $1/\poly(\lambda)$ security loss or required inefficient encryption/decryption operations. We avoid both of these caveats; in doing so, we obtain (to our knowledge) the first plain-model construction of many-time secure $1 \to 2$ unclonable encryption for arbitrary polynomial-length messages, assuming the existence of pseudorandom function-like states (Bartusek and Goldin, arXiv:2605.27647).

The key is a uniformly random non-identity phase-free Pauli on $n$ qubits, and bit $a$ is encrypted as a random $(-1)^a$ eigenstate of that Pauli. Encryption and decryption use $O(n)$ single-qubit operations and $O(n)$ time classical computation; key generation uses only $O(n)$ time classical computation. The scheme is exponentially secure; we prove that the probability that both receivers recover the bit is at most $\frac{1}{2}+\frac{1}{2}\sqrt{{2^n}/({4^n-1})} = \frac{1}{2} + O\left(2^{-n/2}\right).$ By a lower bound due to Broadbent, Culf, and Rochette, this is the best probability bound achievable with $n$-qubit ciphertexts (up to the constant hidden in the $O(\cdot)$).

The main conceptual idea is to leverage, in a precise spectral sense, the balanced commutation-anticommutation structure of the Pauli group. The proof is intricate but completely elementary and makes use of standard spectral bound techniques. The main technical workhorse is a standalone linear-algebraic lemma which we present in its own section: informally, it relates the positivity of two different operators, each capturing the intuition that if the two receivers can individually decrypt unusually often then they must also disagree often.

GPT-5.6 Sol Ultra found this proof in an extended conversation with the author and drafted a preliminary version of this paper. The author is fully accountable for the correctness of this paper. 
\end{abstract}

\section{Introduction}

Unclonable encryption~\cite{Gottesman2003,BroadbentLord2020} asks a quantum ciphertext sent by Alice to enforce a temporal access
rule: before learning the secret key, an adversary (``the cloner'') may process and split the
ciphertext arbitrarily; after the key is revealed, two separated receivers (``Bob'' and ``Charlie'')
should not both be able to recover the plaintext. In this setting, the only honest party is Alice; the cloner, Bob, and Charlie are all adversarial. The baseline success
probability for a one-bit message is $1/2$, since Bob and Charlie may agree in
advance on a uniformly random bit. Alternatively, the cloner could forward Alice's ciphertext to Bob, who can now decrypt and leave Charlie to guess randomly. The goal is to prevent the adversaries from improving non-negligibly over these baseline strategies.

The notion originates in the study of quantum authentication and unclonable
encryption~\cite{Gottesman2003,BroadbentLord2020} and has subsequently been investigated in
oracle, computational, and information-theoretic settings
\cite{BroadbentLord2020,MajenzSchaffnerTahmasbi2021,AnanthEtAl2022}.
It is closely related to monogamy-of-entanglement games
\cite{TomamichelEtAl2013} and, more broadly, to cloning games
\cite{DBLP:conf/crypto/AnanthKL22,PorembaRagavanVaikuntanathan2026}.  Recent work has made tremendous progress on this deceptively simple problem, but to our knowledge, all previous results inherited at least one caveat:
\begin{itemize}
    \item A line of works by Ananth et al. solved this problem assuming either a classical random oracle with quantum query access~\cite{AnanthEtAl2022} or quantum decryption keys~\cite{DBLP:conf/innovations/AnanthKY25}.
    \item Bhattacharyya and Culf developed a decoupling-based approach~\cite{BhattacharyyaCulf2025} that worked in the plain model but incurred a security loss that was inverse-polynomial in $n$.
    \item Later, Bhattacharyya, Broadbent, and Culf~\cite{BhattacharyyaBroadbentCulf2026} used a Haar-based approach in the plain model that attained negligible security loss but required inefficient encryption and decryption operations.
    \item Bartusek and Goldin~\cite{DBLP:journals/corr/abs-2603-11437} solved this problem in the Haar random-oracle model.
\end{itemize}
In this paper we eliminate all of these caveats, showing that one-time, information-theoretically secure, efficient unclonable encryption exists. Our unclonable encryption scheme is as follows. Let
\[
  \mathcal P_n=\{I,X,Y,Z\}^{\ot n}
\]
denote one Hermitian, phase-free representative of every $n$-qubit Pauli
class.  The key is sampled uniformly\footnote{The Pauli $I^{\ot n}$ is excluded from the key space to ensure that the key has both $-1$ and $+1$ eigenspaces. Strictly speaking, this restriction is not necessary; key generation could sample from all of $\mathcal{P}_n$, and in the extremely unlikely event that the key is $I^{\ot n}$, Alice could just ``give up'' and send her bit in the clear. This will only incur a security loss of $O(4^{-n})$.} from
$\Theta_n=\mathcal P_n\setminus\{I^{\ot n}\}$.  To encrypt a bit $a$, we send
a random $(-1)^a$ eigenstate of the key Pauli, or equivalently the maximally mixed state on the $(-1)^a$ eigenspace of the key Pauli.  The
legitimate decryptor measures that Pauli.  The construction is thus finite,
explicit, and efficient.
This construction bears some similarity to one proposed by Botteron et al.~\cite{botteron2026towards}, as we will discuss at the end of this introduction.

Our main theorem is stated below.
The main conceptual idea behind its proof is to leverage the balanced commutation-anticommutation structure of the Pauli group, in a very precise spectral sense: if we encode this structure in a square matrix with $\pm 1$ entries, the resulting matrix will have low spectral norm.
The security proof is intricate but completely elementary and makes use of standard spectral bound techniques.
The main technical workhorse is a standalone linear-algebraic lemma which we present in Section~\ref{sec:separator}. (We defer it to after our main security proof to better motivate its statement.)

\begin{maintheorem}\label{thm:main-intro}
For every $n\geq1$, the Pauli construction is perfectly correct and, against
every information-theoretic splitting attack,
\[
  \Pr[\textnormal{Bob and Charlie decrypt correctly}]
  \leq \frac12+\frac12\sqrt{\frac{2^n}{4^n-1}}.
\]
In particular, the cloning advantage over the trivial $1/2$ strategy
is $O(2^{-n/2})$.
Encryption and
decryption use $O(n)$ single-qubit operations and $O(n)$ time classical computation; key generation uses only $O(n)$ time classical computation.
\end{maintheorem}
\noindent
Some remarks are in order.
Firstly, the $O(2^{-n/2})$ cloning advantage bound is the best achievable for any perfectly correct unclonable encryption scheme with $n$-qubit ciphertexts~\cite[Corollary 22]{broadbent2025optimaluntelegraphableencryptionimplications}.
Secondly, our proof of Main Theorem~\ref{thm:main-intro} also shows that we can take the key space to be any collection of non-identity Paulis of size $2^{n+\omega(\log n)}$ (rather than all $4^n-1$ of them), and still incur a security loss of $\negl(n)$. See Theorem~\ref{thm:security} for a general formulation of our security bound.
Thirdly, our result can be readily upgraded via standard Haar invariance techniques~\cite{MajenzSchaffnerTahmasbi2021,broadbent2025optimaluntelegraphableencryptionimplications,PorembaRagavanVaikuntanathan2026} to give an alternate proof that the Haar scheme of~\cite{BhattacharyyaBroadbentCulf2026} is secure (albeit inefficient), while also making a minor improvement to their proven security loss from $O(2^{-n/8})$ to the optimal $O(2^{-n/2})$.
We present this in Section~\ref{sec:haar} for completeness.

Finally, our result can be upgraded to satisfy many-time security and encrypt arbitrary polynomial-length messages at the necessary expense of making mild computational assumptions, via the compiler of Bartusek and Goldin~\cite{bartusek2026noteboostinguncloneableencryption} building on~\cite{DBLP:conf/tcc/AnanthK21,hiroka2024robust}.
To our knowledge, our work is the first to obtain the following result, due to the aforementioned caveats with previous work.
\begin{corollary}[Via~\cite{bartusek2026noteboostinguncloneableencryption}]\label{cor:bg26}
    Assuming the existence of pseudorandom function-like states and computationally bounded adversaries, there exists many-time secure $1 \to 2$ unclonable encryption for messages of arbitrary polynomial length in the plain model.
\end{corollary}

\paragraph{Comparison with the work of~\cite{botteron2026towards}.}
A very similar construction to ours was considered by Botteron et al.~\cite{botteron2026towards}, but instead of using the entire Pauli group (excluding the identity), they restricted the key space to a pairwise anticommuting subset of the Pauli group.
This immediately limits the construction to security of the form $1/2 + 1/\poly(n)$ because the number of possible keys is at most polynomial, and whether their construction attains this security bound was posed by the authors as a conjecture.
Despite having a similar encryption scheme, our techniques do not appear to imply any useful bounds on the security of their scheme.

The reason is that an essential component of our analysis is leveraging the fact that the Paulis are very balanced between commutation and anticommutation (see Section~\ref{sec:commoverview} for more discussion of this).
Extending our results to cover the construction proposed by~\cite{botteron2026towards} (albeit with a necessarily weaker security bound) is an interesting direction for future work.

\paragraph{Concurrent and independent work.}
Concurrently and independently to our work, Ananth and Sahai~\cite{ananth2026unconditionalunclonableencryption} found and posted the same result.
Their construction and proof are essentially identical to ours, modulo minor differences in presentation and structure.
The only substantive difference is that we give a different proof of the lemma that $\norm{S_BS_C} \leq \epsilon$ (Lemma~\ref{lem:arbitrary-word}, compare with~\cite[Proposition 3.5]{ananth2026unconditionalunclonableencryption}).
Their proof uses a conditional overlap lemma and the pairwise orthogonality of the $n$-qubit Paulis (with respect to the Hilbert-Schmidt inner product); ours uses the Paulis' commutation-anticommutation structure.

\paragraph{Lean 4 formalization.}
To help with verifiability, we release Lean 4 code for our main proofs.
In more detail, the Lean repository proves all correctness and security claims in Main Theorem~\ref{thm:main-intro} and Theorem~\ref{thm:haarmain}, but it does not formally show the efficiency claim in Main Theorem~\ref{thm:main-intro}.
It also provides a modular interface defining unclonable encryption for a bit, allowing for other users to provide their own constructions and security proofs in Lean 4.
It also does not prove Corollary~\ref{cor:bg26} since this would also require formalizing the results of~\cite{bartusek2026noteboostinguncloneableencryption}.
The Lean code was built using Codex GPT-5.6 Sol and can be found at \href{https://github.com/sragavan99/lean-unclonable-enc-paulis}{this GitHub repository.}

\paragraph{Notation.}
All Hilbert spaces are finite-dimensional and complex.  For a linear operator
$A$, write $A^*$ for its adjoint and $\norm{A} = \sqrt{\lambda_{\max}(A^*A)}$ for its spectral norm. For
self-adjoint operators $A,B$, the notation $A\preceq B$ means that $B-A$ is
positive semidefinite.
For a Hermitian $A$ and a state $\rho$, we use $\langle A \rangle_\rho$ to denote the expectation value $\Tr{A\rho}$.
We freely suppress identity operators and tensor
symbols when their registers are clear (see Remark~\ref{remark:registers} for more detail on this).  Density operators have trace one,
and quantum channels are completely positive and trace preserving.  We use
standard facts about Choi representations, Naimark dilation, and finite
dimensional operator theory; see, for example,~\cite{NielsenChuang11,Watrous2018}.

\paragraph{Statement on AI use.}
The result was discovered by OpenAI's GPT-5.6 Sol Ultra in conversation with the author.
The author asked the agents to brainstorm and pursue an array of approaches to proving the theorem.
The agents eventually converged on the Pauli bit as a promising candidate encryption scheme.
In the first few rounds, the agents managed to prove bounds of the form $c+o(1)$ or $c+\negl(n)$ for slowly improving constants $c > 1/2$, for the most part using sum-of-squares bounds based on the NPA hierarchy.
The author explicitly pushed the agents to find a proof using simple spectral bounds and provided feedback on their progress.
In the sixth round, the agents found a complete solution.
The agents drafted the first version of this paper, after which the author significantly cleaned up and reorganized the formal presentation and added informal exposition explaining the ideas.
GPT-5.6 Sol and Claude Fable 5 were both used for proofreading.
Codex GPT-5.6 Sol was used to produce Lean 4 code for our proofs.

\paragraph{Acknowledgements.}
The author thanks Alexander Poremba and Vinod Vaikuntanathan for introducing him to the problem of unclonable encryption and for a wonderful collaboration on this front.
The author additionally thanks Vinod Vaikuntanathan and Aparna Gupte for helpful feedback on drafts of this paper, and Prabhanjan Ananth, Eli Goldin, Yao-Ting Lin, and Amit Sahai for helpful discussions.
The author's access to ChatGPT Pro and Claude Max were respectively supported by the UK AISI alignment project, and Jane Street.

\section{Technical Overview}\label{sec:proofoverview}

The security proof for the Pauli unclonable encryption scheme begins in Section~\ref{sec:reduction} by applying standard transformations~\cite{BroadbentLord2020} to boil proving security down to the following task: let $\Theta$ denote the key space of all non-identity Paulis, and for each $\theta$ let $P_\theta, B_\theta, C_\theta$ respectively denote the corresponding Pauli, Bob's reflection, and Charlie's reflection.
By ``Bob's reflection'', we mean that Bob outputs his guess given the secret key $\theta$ by applying the projective measurement $\{(\id + B_\theta)/2, (\id-B_\theta)/2\}$.
We understand $P_\theta, B_\theta, C_\theta$ to be operators in the same Hilbert space, with the constraint that $P_\theta, B_\phi, C_\psi$ must pairwise commute for any $\theta, \phi, \psi \in \Theta$ (which is enforced by the spatial separation of the parties).
In these terms, the goal is simply to show that:
\begin{equation}\label{eq:overviewtarget}
\Xi := \frac{1}{|\Theta|} \sum_{\theta \in \Theta} \left(\frac{\id + P_{\theta, \reg{A}}^\top B_{\theta, \reg{B}}}{2}\right)\left(\frac{\id + P_{\theta, \reg{A}}^\top C_{\theta, \reg{C}}}{2}\right) \preceq \left(\frac{1}{2} + \negl(n)\right)\id.
\end{equation}
The subscripts $\reg{A}, \reg{B}, \reg{C}$ above indicate which register each of these operators acts on.
Going forward, we will suppress these subscripts and understand $P_\theta, B_\theta, C_\theta$ as operators on the combined Hilbert space, carrying forward the constraint that any two of $P_\theta, B_\phi, C_\psi$ must commute (see Remark~\ref{remark:registers} for more on this).
The intuition is that the $\frac{\id + P_\theta^\top B_\theta}{2}$ projector checks whether Bob guesses correctly and the $\frac{\id + P_\theta^\top C_\theta}{2}$ projector checks whether Charlie guesses correctly.
The adversaries' strategy can be encoded with the existing unitaries $\{B_\theta, C_\theta\}$ and a state $\rho$ on this same Hilbert space.
For brevity, write $\Pi_\theta = \left(\frac{\id + P_\theta^\top B_\theta}{2}\right)\left(\frac{\id + P_\theta^\top C_\theta}{2}\right)$.

It might be tempting to try and prove the stronger two-sided bound that $\norm{\Xi - \id/2} \leq \negl(n)$, but this is false!
If Bob always outputs 0 and Charlie always outputs 1, then their winning probability is 0 so $\langle \Xi \rangle_\rho$ will be 0.
With that in mind, we pursue the one-sided objective stated in~\eqref{eq:overviewtarget}.

\subsection{Operators and Interpretations (Section~\ref{sec:newops})}\label{sec:overviewnewops}

All techniques discussed in this overview can be interpreted as proving an inequality between three operators:
\begin{itemize}
  \item $S_B = \frac{1}{|\Theta|} \sum_\theta P_\theta^\top B_\theta$: this is an average of reflections that check whether Bob's guess is correct.
  If $p_B$ is the probability that Bob guesses correctly (averaged across $\theta$), we have $\langle S_B \rangle_\rho = 2p_B - 1$.

  \item $S_C = \frac{1}{|\Theta|} \sum_\theta P_\theta^\top C_\theta$: this is defined analogously for Charlie.
  Let $p_C$ denote the probability that Charlie guesses correctly.

  \item $D = \frac{1}{|\Theta|} \sum_\theta \frac{\id - B_\theta C_\theta}{2}$: this is an average of projections that check whether Bob and Charlie's guesses disagree.
  If $p_{\neq}$ is the probability that they disagree, we have $\langle D \rangle_\rho = p_{\neq}$.
\end{itemize}
In these terms, it is straightforward to see that:
\begin{equation}\label{eq:xialtoverview}
\Xi = \frac{1}{2} \id + \frac{1}{4}(S_B+S_C) - \frac{1}{2}D.
\end{equation}
This can either be checked with direct algebraic manipulation, or intuited by noting that $\frac{S_B+S_C}{4}$ is capturing the bias beyond $1/2$ of the sums of Bob's and Charlie's individual guessing probabilities.
This overcounts cases where exactly one of Bob and Charlie guesses correctly, which we account for by subtracting off $\frac{1}{2}D$.
In general, the above operational interpretations will be extremely useful for understanding the techniques which we present.
At this point, one can also intuit the following bounds (which we prove rigorously in the body).
The limitation of these bounds is that they can all be deduced without considering how the operators for different keys $\theta \neq \theta'$ interact with each other.
\begin{align}
  0 &\preceq D \preceq \id \qquad\text{(``the probability of disagreement is always between 0 and 1'');}\label{item:Dproj} \\
  D &\succeq \pm \frac{S_B-S_C}{2} \qquad\text{(``if Bob succeeds and Charlie does not, then their answers disagree'').} \label{item:diff} \\
\end{align}

\subsection{Leveraging the Pauli Commutation-Anticommutation Structure (Lemma~\ref{lem:arbitrary-word})}\label{sec:commoverview}
The main conceptual ingredient of our security analysis is to exploit the balanced commutation-anticommutation behavior of the Pauli group.
In more detail, let $\Gamma$ be the square matrix of dimension $4^n-1$ with $\pm 1$ entries that encodes the Pauli commutation structure i.e., we have $P_\theta P_\phi = \Gamma_{\theta, \phi} P_\phi P_\theta$ for all $\phi, \theta$.
A key component of our bound will be the straightforward estimate that $\norm{\Gamma} \leq 2^n$ (Lemma~\ref{lem:symplectic-matrix}); this essentially follows from the fact that $\Gamma$ is a submatrix of a scaled-up and column-permuted Hadamard matrix.
So a feature of the all-Pauli scheme that works in our favor is that Paulis have a commutation-anticommutation pattern that has nice spectral behavior.

We view Lemma~\ref{lem:symplectic-matrix} discussed above as a spectral formulation of the notion that $\Gamma$ is ``balanced.''\footnote{As discussed at the end of the introduction, this is quite different from the Pauli construction considered by~\cite{botteron2026towards}, where the proposed scheme worked with a pairwise anticommuting set of Paulis (so $\Gamma$ would have 1's on the diagonal and $-1$'s off the diagonal and have exponentially lower dimension).}
An example of an unbalanced $\Gamma$ would be the all 1's matrix, which has larger spectral norm $4^n-1$.
Using this commutation structure alone, it is straightforward to prove that:
\begin{equation}\label{eq:overviewarbitraryword}
  \norm{S_BS_C} \leq \epsilon \qquad\text{(Lemma~\ref{lem:arbitrary-word})},
\end{equation}
where $\epsilon = \norm{\Gamma}/|\Theta| = \negl(n)$.
To see why, let us expand out $S_BS_C$:
\[
    S_BS_C = \frac{1}{|\Theta|^2} \sum_{\theta, \phi} P_\theta^\top B_\theta P_\phi^\top C_\phi = \frac{1}{|\Theta|^2} \sum_{\theta, \phi} C_\phi (P_\phi P_\theta)^\top B_\theta.
\]
This is an average of $|\Theta|^2$ unitaries, so it is not yet clear why its spectral norm should be $\ll 1$.
The trick is to commute $P_\phi$ and $P_\theta$ past each other, yielding:
\[
    S_BS_C = \frac{1}{|\Theta|^2} \sum_{\theta, \phi} \underbrace{(C_\phi P_\phi^\top)}_{:= U_\phi} \Gamma_{\phi, \theta} \underbrace{(P_\theta^\top B_\theta)}_{:= V_\theta}.
\]
Now we have a \emph{signed} average of $|\Theta|^2$ unitaries $\{U_\phi V_\theta: \phi, \theta \in \Theta\}$! This is the quadratic form of $\Gamma$ applied to the matrices $\{U_\phi\}$ on the left and $\{V_\theta\}$ on the right. So the fact that $\Gamma$ has low spectral norm now forces this average to destructively interfere and give us the desired negligible bound on $\norm{S_BS_C}$.

\subsection{Introducing a Regularized Disagreement Conditioner (Corollary~\ref{cor:cross-resolvent})}\label{sec:conditioner}
Let us for the moment pretend that (a) $S_B$ and $S_C$ commute (so $S_BS_C$ is Hermitian); and (b) we have $\langle S_BS_C\rangle_\rho = \langle S_B \rangle_\rho \langle S_C \rangle_\rho$ (i.e., Bob's and Charlie's success probabilities across keys are uncorrelated in some sense).
With these assumptions,~\eqref{eq:overviewarbitraryword} would admit the very clean interpretation that
\begin{equation}\label{eq:pbpc}
  |(2p_B-1)(2p_C-1)| \leq \epsilon.
\end{equation}
In this case, the above would imply that $p_B, p_C$ cannot simultaneously be $> 1/2 + \sqrt{\epsilon}/2$ and we would immediately be done.

What goes wrong when we account for correlations between Bob's guess and Charlie's guess? In this case,~\eqref{eq:pbpc} need not be true.
Consider the following very simple strategy: the cloner flips a coin and accordingly forwards Alice's ciphertext to either Bob or Charlie.
The other receiver will guess randomly.
For this strategy, we have $p_B = p_C = 3/4$, contradicting~\eqref{eq:pbpc}.
The reason that Bob and Charlie still only win with probability $1/2$ is that the cases where Bob guesses incorrectly and Charlie guesses incorrectly are completely disjoint (i.e., there is a strong anticorrelation), so the two $1/4$'s add up to give a failure probability of $1/2$.
So we need to somehow adapt~\eqref{eq:overviewarbitraryword} to additionally penalize strategies where Bob and Charlie disagree often.

Fortunately for us, we have an operator $D$ that exactly captures the subspace where Bob and Charlie disagree.
A concrete hope would be to prove a statement of the form:
\begin{equation}\label{eq:wishful}
    \norm{S_BD^{-1}S_C} \leq \negl(n).
\end{equation}
In the above, $D^{-1}$ acts as a ``disagreement conditioner'' that amplifies subspaces where Bob and Charlie are likely to agree but not subspaces where Bob and Charlie are likely to disagree (and therefore lose).
This effectively focuses the estimate on the subspace we really need to worry about which is that where Bob and Charlie agree.

Unfortunately,~\eqref{eq:wishful} cannot possibly be true; $D$ need not be invertible, as there could be a state for which Bob and Charlie agree with probability 1.
However, we can prove something very close, namely:
\begin{equation}\label{eq:overviewresolvent}
  \norm{S_B (D+\delta \id)^{-1} S_C}\leq \epsilon/\delta \quad \forall \delta > 0 \qquad\text{(Corollary~\ref{cor:cross-resolvent})}.
\end{equation}
We think of $(D+\delta \id)^{-1}$ as a ``regularized disagreement conditioner''; the $\delta \id$ term acts as a ``regularizer'' to ensure that no direction gets blown up by more than $1/\delta$ and we accordingly pay a $1/\delta$ penalty on the RHS.
This penalty can be absorbed while retaining negligibility by setting $\delta = \sqrt{\epsilon}$.

The proof of~\eqref{eq:overviewresolvent} proceeds by taking a power series expansion of $(D+\delta \id)^{-1}$ and bounding every term that appears using~\eqref{eq:overviewarbitraryword} (Lemma~\ref{lem:arbitrary-word}).

\subsection{Linear-Algebraic Step (Section~\ref{sec:separator})}\label{sec:separatoroverview}
It turns out that our heuristic intuition about the relevance of~\eqref{eq:overviewresolvent} can actually be turned into a rigorous security proof!
Recall from~\eqref{eq:xialtoverview} that our task is to show that:
\begin{equation}\label{eq:overviewdone}
  S_B+S_C - 2D \preceq \negl(n) \cdot \id.
\end{equation}
This bears some qualitative similarity to~\eqref{eq:overviewresolvent}; both say that ``$D$ is large relative to $S_B$ and $S_C$'', or more qualitatively that if Bob and Charlie individually win unusually often then they must also disagree often.
In fact, we can formally show that~\eqref{eq:overviewresolvent} implies~\eqref{eq:overviewdone}; this is the content of Theorem~\ref{thm:separator}.

The proof uses some simple but non-obvious manipulations with $2 \times 2$ block matrices based on Schur complements; equivalently, it can be written as a matrix-valued rational sum-of-squares certificate.

\section{Preliminaries}

\subsection{Pauli Symplectic Space}

Let $V_n=\F_2^n\times\F_2^n$.
For $v=(x,z)\in V_n$, let $P_v$ be the phase-free Hermitian Pauli obtained by
taking the tensor product of $I,X,Y,Z$ prescribed by the pairs
$(x_i,z_i)$.  Equivalently, one may use the phase convention
$P_{(x,z)}=i^{|x\wedge z|}X^xZ^z$, where $|x\wedge z|$ is the integer
number of coordinates on which $x_i=z_i=1$.
Every $P_v$ is a reflection (meaning $P_v^2 = \id$).

\begin{definition}[Pauli symplectic form]
    For $v = (x, z), w = (x', z') \in V_n$, we define the symplectic form to be:
    \[
      [v,w]=\langle x, z'\rangle+\langle z, x'\rangle\in\F_2.
    \]
\end{definition}

\begin{lemma}[Pauli commutation relations]\label{lem:paulibasic}
    For any $v, w \in V_n$, we have:
    \begin{equation}\label{eq:pauli-basic}
      P_vP_w=(-1)^{[v,w]}P_wP_v,\qquad \Tr{P_vP_w}=2^n\one[v=w].
    \end{equation}
\end{lemma}
\begin{proof}
    Writing $v = (x, z)$ and $w = (x', z')$, we have:
    $$P_vP_w = i^{|x \land z| + |x' \land z'|} X^x Z^z X^{x'} Z^{z'} = i^{|x \land z| + |x' \land z'|} \cdot (-1)^{\langle x', z \rangle} X^{x+x'} Z^{z+z'}.$$
    (This already implies that $\Tr{P_vP_w} = 0$ unless $(x, z) = (x', z')$, in which case $P_vP_w$ is the identity since the Paulis are reflections.)
    By symmetry, we also have:
    $$P_wP_v = i^{|x \land z| + |x' \land z'|} \cdot (-1)^{\langle x, z' \rangle} X^{x+x'} Z^{z+z'}.$$
    Comparing the above two equalities implies the conclusion.
\end{proof}

\begin{definition}[Symplectic character matrix]
    Define $\Gamma^{\mathrm{Pauli}} \in \{-1, 1\}^{(4^n-1) \times (4^n-1)}$ to be the matrix with rows and columns indexed by elements of $V_n \backslash \{(0, 0)\}$.
    The entry in row $v$ and column $w$ is $\Gamma^{\mathrm{Pauli}}_{v, w} = (-1)^{[v, w]}$.
\end{definition}

\begin{lemma}[Pauli symplectic structure is spectrally balanced]\label{lem:symplectic-matrix}
    We have $\norm{\Gamma^{\mathrm{Pauli}}} = 2^n$.
    (We remark that we only need the upper bound of $2^n$, but we prove tightness for completeness.)
\end{lemma}
\begin{proof}
    Unpacking $v = (x, z)$ and $w = (x', z')$, the corresponding entry of $\Gamma^{\mathrm{Pauli}}$ is $\Gamma^{\mathrm{Pauli}}_{(x, z); (x', z')} = (-1)^{\langle x, z' \rangle + \langle z, x' \rangle}$.
    Let us permute columns by swapping the roles of $x', z'$ to obtain a new matrix $\Gamma'$, then it suffices to show that $\norm{\Gamma'} = 2^n$.
    We have $\Gamma'_{(x, z); (z', x')} = (-1)^{\langle x, z' \rangle + \langle z, x' \rangle}$, so it is apparent that $\frac{1}{2^n} \Gamma'$ is a submatrix of the $2^{2n} \times 2^{2n}$ normalized Hadamard matrix (removing the row and column corresponding to $(0, 0)$), which is unitary.
    It immediately follows that $\norm{\Gamma'/2^n} \leq \norm{H} = 1 \Rightarrow \norm{\Gamma'} \leq 2^n$.
    To see that this is an equality, note that:
    \begin{align*}
      H &= \begin{bmatrix} \frac{1}{2^n} & \frac{1}{2^n}\one^\top \\ \frac{1}{2^n}\one & \frac{1}{2^{n}}\Gamma' \end{bmatrix} \\
      \Rightarrow 2^{2n} \id = (2^{n}H)^2 &= \begin{bmatrix} 1 & \one^\top \\ \one & \Gamma' \end{bmatrix}^2 \\
      &= \begin{bmatrix} 1 & \one^\top \\ \one & \Gamma' \end{bmatrix} \begin{bmatrix} 1 & \one^\top \\ \one & \Gamma' \end{bmatrix} \\
      &= \begin{bmatrix} \cdot & \cdot \\ \cdot & \one\one^\top + (\Gamma')^2\end{bmatrix} \\
      \Rightarrow 2^{2n} \id_{2^{2n}-1} &= \one\one^\top + (\Gamma')^2.
    \end{align*}
    So for a nonzero vector $\vec{v} \in \one^\perp \subseteq \CC^{2^{2n}-1}$, we have $2^{2n} \vec{v} = (\Gamma'^2) \vec{v}$, forcing $\norm{\Gamma'^2} \geq 2^{2n} \Rightarrow \norm{\Gamma'} \geq 2^{n}$.
\end{proof}

\subsection{Choi-Jamiołkowski Isomorphism}\label{sec:CI}

Let $\algo H_{\reg A}$ be a $d$-dimensional Hilbert space with an orthonormal basis denoted by $\{\ket{1},\dots,\ket{d}\}$. Let $\ket{\Omega} = \sum_{i \in [d]} \ket{i} \otimes \ket{i}$ be the vectorization of the identity $\id_d = \sum_{i \in [d]} \proj{i}$. Then, the normalized Choi-Jamiołkowski isomorphism $J(\Phi) \in \linear(\algo H_{\reg B} \otimes \algo H_{\reg A'})$ with respect to a CPTP map of the form $\Phi: \linear(\algo H_{\reg A}) \rightarrow \linear(\algo H_{\reg B})$ is defined as 
$$
J(\Phi) =  \frac{1}{d} (\Phi_{\reg A \rightarrow \reg B} \otimes \id_{\reg{A}'})(\proj{\Omega}) = \frac{1}{d}\sum_{i,j \in [d]} \Phi(\ketbra{i}{j}) \otimes \ketbra{i}{j}.
$$
Note that $J(\Phi)$ is a valid quantum state.
We use the following well known fact (copied here from~\cite[Corollary 2.2]{PorembaRagavanVaikuntanathan2026} while taking our different normalization convention into account).
\begin{lemma}\label{lem:CIprojector}
Let $\Phi: \linear(\algo H_{\reg A}) \rightarrow \linear(\algo H_{\reg B})$ be any linear map. Then, 
    for any Hermitian operators $\vec{P} \in \mathrm{L}(\algo{H}_{\reg B})$ and $\vec{Q} \in \mathrm{L}(\algo{H}_{\reg A})$, it holds that $$\Tr{\vec{P} \Phi(\vec{Q})} = d \cdot \Tr{\left(\vec{P} \otimes \bar{\vec{Q}}\right)J(\Phi)}.$$
\end{lemma}

\subsection{One-Time Secure Unclonable Encryption}

Consider a key-indexed pair of states $\rho_{\theta,0},\rho_{\theta,1}$.
The plaintext $a\leftarrow\bit$ and key $\theta$ are uniform.  The cloner
first applies a key-independent channel
\[
  \Phi:A_0\longrightarrow B\ot C
\]
to the ciphertext.  The key is then revealed, after which Bob and Charlie use
arbitrary key-dependent binary POVMs
$\{D^B_{\theta,a}\}_{a\in\bit}$ and
$\{D^C_{\theta,a}\}_{a\in\bit}$.  Their joint success probability is
\begin{equation}\label{eq:security-game}
  \frac1{2|\Theta|}\sum_{\theta\in\Theta}\sum_{a\in\bit}
  \Tr{(D^B_{\theta,a}\ot D^C_{\theta,a})
  \Phi(\rho_{\theta,a})}.
\end{equation}
The scheme is one-time unclonably secure if this is at most
$1/2+\negl(n)$ for every attack.

\section{The Unclonable Encryption Scheme}\label{sec:construction}

To help highlight the important aspects of our proof, we formulate our construction in slightly more general terms than our Pauli instantiation.

\begin{construction}\label{con:pauli}
Fix $n \geq 1$ and put $d = 2^n$.
Let $\{Q_\theta: \theta \in \Theta_n\}$ be any collection of traceless Hermitian reflections of dimension $2^n$.
Let $M = |\Theta_n|$.
We will often abuse notation and let $\Theta = \Theta_n$.
\begin{itemize}
  \item Key generation: sample a uniformly random $\theta \gets \Theta$.
  \item To encrypt
  $a\in\bit$, output
  \begin{equation}\label{eq:encryption-state}
    \rho_{\theta,a}
    =\frac1d\bigl(\id+(-1)^aQ_\theta\bigr).
  \end{equation}
  An equivalent formulation is to send a uniformly random $(-1)^a$ eigenstate of $Q_\theta$.
  \item To decrypt, measure the reflection $Q_\theta$ and output the bit whose sign
  $(-1)^a$ equals the measured eigenvalue.
\end{itemize}
\end{construction}

\begin{lemma}[Correctness]\label{lem:correctness}
For every $\theta\in\Theta$ and $a\in\bit$,
$\rho_{\theta,a}$ is a density operator.  Construction~\ref{con:pauli} is
perfectly correct.
\end{lemma}

\begin{proof}
Since $Q_\theta$ is a Hermitian reflection,
\[
  E_{\theta,a}=\frac12(\id+(-1)^aQ_\theta)
\]
is the orthogonal projector onto its $(-1)^a$ eigenspace.  Tracelessness of
$Q_\theta$ gives $\Tr{E_{\theta,a}}=d/2$, so
$\rho_{\theta,a}=(2/d)E_{\theta,a}$ is positive semidefinite and has trace
one.  Measuring $\{E_{\theta,0},E_{\theta,1}\}$ on a state supported on
$E_{\theta,a}$ returns $a$ with probability one.
\end{proof}

\subsection{The Pauli Instantiation}

Our main instantiation will take $\Theta_n = V_n \backslash \{(0, 0)\}$ to be the collection of all non-identity Paulis, and $Q_\theta = P_\theta$.
Lemma~\ref{lem:paulibasic} tells us that these are traceless Hermitian reflections and hence we get correctness.

To see the efficiency of this instantiation, let the key be $(\theta_1, \ldots, \theta_n)$ and write $S=\{i:\theta_i\neq I\}$.  To encrypt, sample one-qubit eigenvalue signs $\{\sigma_j \in \{\pm 1\}: j \in S\}$ uniformly subject to $\prod_{j \in S} \sigma_j = (-1)^a$, prepare the corresponding one-qubit Pauli eigenstates, and
prepare the maximally mixed state on coordinates outside $S$.  The resulting
mixture is exactly the maximally mixed state in~\eqref{eq:encryption-state}.
Decryption measures the single-qubit Paulis on $S$ and takes the parity of the
outcomes.  Sampling a nonzero Pauli label, encrypting, and decrypting all take
$O(n)$ elementary operations.

\section{Proof of Security}\label{sec:fullsecurity}

We begin by stating the general security theorem that we will prove in this section:

\begin{theorem}[General unclonability]\label{thm:security}
Consider the construction of Section~\ref{sec:construction}, and assume additionally that there exists a matrix $\Gamma \in \{-1, 1\}^{|\Theta| \times |\Theta|}$ such that:
\begin{equation}\label{eq:commdef}
Q_\phi Q_\theta = \Gamma_{\phi, \theta} Q_\theta Q_\phi \quad \forall \theta, \phi \in \Theta.
\end{equation}
Then for every splitting attack in the game of
Section~\ref{sec:construction},
\begin{equation}\label{eq:main-bound}
  \Pr[\textnormal{Bob and Charlie decrypt correctly}]
  \leq
  \frac12+\frac12\sqrt{\frac{\norm{\Gamma}}{|\Theta|}}.
\end{equation}
\end{theorem}

\subsection{Ricochet Reduction}\label{sec:reduction}

First, we apply some standard manipulations~\cite{BroadbentLord2020} based on Naimark dilation and the EPR ricochet property to reduce security to a simpler linear-algebraic task.
To this end, fix an arbitrary splitting channel and arbitrary key-dependent local binary
POVMs.
By Naimark's dilation theorem, we may assume without loss of generality that Bob's measurement is determined by an ensemble of binary projective measurements $\{(D^B_{\theta, 0}, D^B_{\theta, 1}): \theta \in \Theta\}$, and similarly for Charlie.
(This was shown in~\cite{TomamichelEtAl2013}, so we do not reproduce a proof here.)
We succinctly encode these in the operators:
\[
  B_\theta=D^B_{\theta,0}-D^B_{\theta,1},
  \qquad
  C_\theta=D^C_{\theta,0}-D^C_{\theta,1}.
\]
It is straightforward to see that $B_\theta, C_\theta$ are Hermitian and that $B_\theta^2 = C_\theta^2 = \id$ for all $\theta \in \Theta$.
Since these reflections act on distinct receiver
registers we immediately have commutation:
\begin{equation}\label{eq:cross-commutation}
  B_\theta C_\phi = C_\phi B_\theta
  \qquad\text{for all $\theta,\phi$.}
\end{equation}
Note that there is no commutation assumption among two Bob reflections or two Charlie reflections.

The following lemma can be interpreted as switching to the following experiment: Alice and the cloner initially share $n$ EPR pairs.
The cloner then acts arbitrarily on their end of the state and forwards the results to Bob and Charlie.
Bob and Charlie make their measurements as usual, and Alice measures with the projectors $\frac{1}{2}(\id \pm Q_\theta^\top)$.
The adversaries win if all three measurement results agree.
In the following, all stated operators will be over some common dimension $\tilde{d}$.

\begin{lemma}[Winning operator]\label{lem:winning-operator}

The success probability of the fixed cloning attack is $\omega = \Tr{\Xi J(\Phi)}$, where $J(\Phi)$ is the normalized Choi state as defined in Section~\ref{sec:CI}, and:
\begin{equation}\label{eq:Xi-Pi}
  \Xi=\frac1M\sum_{\theta\in\Theta_n}\Pi_\theta,
  \qquad
  \Pi_\theta=\frac14(\id+Q_\theta^\top B_\theta)
                         (\id+Q_\theta^\top C_\theta).
\end{equation}
Moreover, every $\Pi_\theta$ is an orthogonal projection, and hence
$0\preceq\Xi\preceq\id$.
\end{lemma}

\begin{proof}
We begin from~\eqref{eq:security-game} and use Lemma~\ref{lem:CIprojector}:
\begin{align*}
    \omega &= \frac1{2|\Theta|}\sum_{\theta\in\Theta}\sum_{a\in\bit}
  \Tr{\left(\frac{\id + (-1)^a B_\theta}{2}\right) \left(\frac{\id + (-1)^a C_\theta}{2}\right)
  \Phi\left(\frac1d\bigl(\id+(-1)^aQ_\theta\bigr)\right)} \\
  &= \frac1{2|\Theta|}\sum_{\theta\in\Theta}\sum_{a\in\bit}
  \Tr{\left(\frac{\id + (-1)^a B_\theta}{2}\right) \left(\frac{\id + (-1)^a C_\theta}{2}\right)
  \left(\id+(-1)^aQ_\theta^\top\right)J(\Phi)} \\
  &= \frac1{4|\Theta|}\sum_{\theta\in\Theta}
  \Tr{\left(\id + B_\theta C_\theta + B_\theta Q_\theta^\top + C_\theta Q_\theta^\top\right)J(\Phi)} \\
  &= \frac{1}{|\Theta|} \sum_{\theta \in \Theta} \Tr{\Pi_\theta J(\Phi)} \\
  &= \Tr{\Xi J(\Phi)},
\end{align*}
as claimed.
To see the latter claims, note that $Q_\theta^\top, B_\theta, C_\theta$ are commuting reflections so it follows that $\frac{\id + Q_\theta^\top B_\theta}{2}, \frac{\id + Q_\theta^\top C_\theta}{2}$ are commuting projectors, so their product $\Pi_\theta$ is also a projector.
The claim that their average $\Xi$ is PSD and has all eigenvalues $\leq 1$ immediately follows.
\end{proof}

\begin{remark}\label{remark:registers}
    To clarify, we are abusing notation slightly: what we should really do is define $Q_\theta, B_\theta, C_\theta$ locally on Alice's, Bob's, and Charlie's registers, then define:
    \begin{align*}
        \widehat{Q}_\theta &= Q_{\theta, \reg{A}} \otimes \id_{\reg{B}} \otimes \id_{\reg{C}} \\
        \widehat{B}_\theta &= \id_{\reg{A}} \otimes B_{\theta, \reg{B}} \otimes \id_{\reg{C}} \\
        \widehat{C}_\theta &= \id_{\reg{A}} \otimes \id_{\reg{B}} \otimes C_{\theta, \reg{C}},
    \end{align*}
    and then work with the operators $\widehat{Q}_\theta, \widehat{B}_\theta, \widehat{C}_\theta$.
    The abuse of notation is that we take $Q_\theta$ to mean $\widehat{Q}_\theta$, and similarly for $B_\theta, C_\theta$.
\end{remark}

\subsection{Setup}\label{sec:newops}

We refer the reader to Section~\ref{sec:overviewnewops} for an intuitive discussion of the operators defined here and their properties.

\begin{definition}
    Define the following Hermitian operators:
    \begin{align}
      S_B&=\frac1M\sum_\theta Q_\theta^\top B_\theta,
      &S_C&=\frac1M\sum_\theta Q_\theta^\top C_\theta,
      &D&=\frac{1}{M} \sum_\theta \frac{\id-B_\theta C_\theta}{2}.
    \end{align}
    Also, define the parameter 
    $$\epsilon := \epsilon_n =\frac{\norm{\Gamma}}{M}.$$
\end{definition}

\begin{lemma}
\label{lem:game-decomposition}
The winning operator satisfies
\begin{equation}\label{eq:Xi-YR}
  \Xi=\frac12\id+\frac14(S_B+S_C) - \frac{1}{2}D.
\end{equation}
Furthermore,
\begin{equation}\label{eq:R-constraints}
  0\preceq D \preceq\id,
  \qquad D \succeq \frac{1}{2}(S_B-S_C), \qquad D \succeq \frac{1}{2}(S_C-S_B).
\end{equation}
\end{lemma}

\begin{proof}
Expanding~\eqref{eq:Xi-Pi} gives:
\begin{align*}
    \Xi &= \frac{1}{M} \sum_{\theta \in \Theta} \frac14(\id+Q_\theta^\top B_\theta)
                         (\id+Q_\theta^\top C_\theta) \\
        &= \frac{1}{4} \id + \underbrace{\frac{1}{4M} \sum_{\theta \in \Theta} (Q_\theta^\top B_\theta + Q_\theta^\top C_\theta)}_{= (S_B+S_C)/4} + \underbrace{\frac{1}{4M} \sum_{\theta \in \Theta} B_\theta C_\theta}_{= \id/4 - D/2} \\
        &= \frac{1}{2} \id + \frac{1}{4} (S_B+S_C) - \frac{1}{2} D,
\end{align*}
establishing the first identity.
To see the positivity constraints, first note that $D$ is an average over $\theta$ of the projector $\frac{\id - B_\theta C_\theta}{2}$, so $0 \preceq D \preceq \id$ is immediate.
Secondly, for any $s \in \{\pm 1\}$ and any $\theta \in \Theta$, the projectors $\frac{1}{2}(\id + sQ_\theta^\top B_\theta)$ and $\frac{1}{2}(\id - s Q_\theta^\top C_\theta)$ commute, so we have:
\begin{align*}
    0 & \preceq (\id + s Q_\theta^\top B_\theta)(\id - s Q_\theta^\top C_\theta) \\
    \Rightarrow 0 &\preceq \id + \frac{s}{M} \sum_{\theta \in \Theta} Q_\theta^\top B_\theta - \frac{s}{M} \sum_{\theta \in \Theta} Q_\theta^\top C_\theta - \frac{1}{M} \sum_{\theta \in \Theta} B_\theta C_\theta \text{ (averaging over $\theta$)}  \\
    &= 2D + s(S_B-S_C) \\
    \Rightarrow \frac{s}{2}(S_C-S_B) &\preceq D,
\end{align*}
as desired.
\end{proof}

\subsection{Bounds Using Pauli Commutation-Anticommutation Structure}\label{sec:word-bound}

As discussed in Section~\ref{sec:commoverview}, the main property of the reflections $\{Q_\theta\}$ that we will leverage is that their commute--anticommute sign matrix $\Gamma \in \{-1, 1\}^{|\Theta| \times |\Theta|}$ has bounded norm (as we verified for the case of Paulis in Lemma~\ref{lem:symplectic-matrix}).
We next turn this norm estimate into a bound for $S_BS_C$:

\begin{lemma}\label{lem:arbitrary-word}
We have:
\begin{equation}\label{eq:arbitrary-word}
  \norm{S_B S_C}\leq \epsilon = \norm{\Gamma}/M.
\end{equation}
\end{lemma}

\begin{proof}
Recall that we can commute any of the relevant unitaries past each other, except for two Bob unitaries or two Charlie unitaries.
In the case of two encoding reflections $Q_\theta$, we can commute them but we must pay a correction factor as per~\eqref{eq:commdef}.
Bearing this in mind, we compute:
\begin{align*}
    S_B S_C &= \frac{1}{M^2} \sum_{\theta, \phi \in \Theta} (Q_\theta^\top B_\theta) (Q_\phi^\top C_\phi) \\
    &= \frac{1}{M^2} \sum_{\theta, \phi \in \Theta} C_\phi Q_\theta^\top Q_\phi^\top B_\theta \\
    &= \frac{1}{M^2} \sum_{\theta, \phi \in \Theta} C_\phi (Q_\phi Q_\theta)^\top B_\theta \\
    &= \frac{1}{M^2} \sum_{\theta, \phi \in \Theta} \Gamma_{\phi, \theta} C_\phi (Q_\theta Q_\phi)^\top B_\theta \\
    &= \frac{1}{M^2} \sum_{\theta, \phi \in \Theta} \Gamma_{\phi, \theta} (C_\phi Q_\phi^\top) (Q_\theta^\top B_\theta).
\end{align*}
To help rewrite the above expression, let $M_C \in \CC^{\tilde{d} \times \tilde{d}|\Theta|}$ be a block matrix with $|\Theta|$ block columns indexed by $\phi \in \Theta$. The block column corresponding to $\phi$ will be $C_\phi Q_\phi^\top$.
Similarly define $M_B \in \CC^{\tilde{d}|\Theta| \times \tilde{d}}$ have row blocks indexed by $\theta \in \Theta$ and containing $Q_\theta^\top B_\theta$.
Finally, let $M_\Gamma \in \CC^{\tilde{d}|\Theta| \times \tilde{d}|\Theta|}$ be defined to be:
$$M_\Gamma := \Gamma \otimes \id_{\tilde{d}}.$$
In words, $M_\Gamma$ has block rows and columns indexed by $\phi$ and $\theta$ respectively. The corresponding $\tilde{d} \times \tilde{d}$ block is $\Gamma_{\phi, \theta}$ times the identity matrix.
In these terms, the outcome of the above computation can be rewritten as:
\begin{align*}
    S_BS_C &= \frac{1}{M^2} M_C M_\Gamma M_B \\
    &= \frac{1}{M^2} M_C (\Gamma \otimes \id_{\tilde{d}}) M_B \\
    \Rightarrow \norm{S_B S_C} &= \frac{1}{M^2} \norm{M_C (\Gamma \otimes \id_{\tilde{d}}) M_B} \\
    &\leq \frac{1}{M^2} \norm{M_C} \cdot \norm{\Gamma} \cdot \norm{M_B}.
\end{align*}
We bound each term of the RHS above as follows:
\begin{itemize}
    \item For any unit $\ket{\psi}$, the vector $M_B \ket{\psi}$ will be the concatenation of $Q_\theta^\top B_\theta \ket{\psi}$ for every $\theta \in \Theta$.
    By unitarity, each of these chunks has unit norm, so the total $\ell_2$ norm will always be exactly $\sqrt{|\Theta|} = \sqrt{M}$.
    Consequently, we have $\norm{M_B} = \sqrt{M}$.
    \item We may similarly obtain $\norm{M_C} = \sqrt{M}$.
\end{itemize}
Putting these together yields:
$$\norm{S_B S_C} \leq \frac{1}{M^2} \cdot M \cdot \norm{\Gamma} = \frac{\norm{\Gamma}}{M} = \epsilon,$$
as desired.
\end{proof}

In the next two corollaries, we build on Lemma~\ref{lem:arbitrary-word} using power series techniques to prove a ``regularized-disagreement-conditioned'' version of its conclusion.
See Section~\ref{sec:conditioner} for a discussion of this ``regularized conditioning'' idea.

\begin{corollary}\label{cor:Z-moments}
    For every integer $\ell \geq 0$, we have:
    \begin{equation}\label{eq:Z-moment-bound}
  \norm{S_B(\id - 2D)^\ell S_C}\leq\epsilon.
\end{equation}
\end{corollary}

\begin{proof}
We compute:
\begin{align*}
    (\id-2D)^\ell &= \frac{1}{M^\ell}\left(\sum_\theta B_\theta C_\theta\right)^\ell \\
    &= \frac{1}{M^\ell} \sum_{\theta_1, \ldots, \theta_\ell \in \Theta} B_{\theta_1} C_{\theta_1} \ldots B_{\theta_\ell} C_{\theta_\ell} \\
    &= \frac{1}{M^\ell} \sum_{\theta_1, \ldots, \theta_\ell \in \Theta} B_{\theta_1} \ldots B_{\theta_\ell} C_{\theta_1} \ldots C_{\theta_\ell} \\
    \Rightarrow S_B(\id-2D)^{\ell} S_C &= \frac{1}{M^\ell} \sum_{\theta_1, \ldots, \theta_\ell \in \Theta} S_B B_{\theta_1} \ldots B_{\theta_\ell} C_{\theta_1} \ldots C_{\theta_\ell} S_C \\
    \Rightarrow \norm{S_B (\id-2D)^{\ell} S_C} &\leq \frac{1}{M^\ell} \sum_{\theta_1, \ldots, \theta_\ell \in \Theta} \norm{S_B B_{\theta_1} \ldots B_{\theta_\ell} C_{\theta_1} \ldots C_{\theta_\ell} S_C} \quad\text{(triangle inequality)} \\
    &= \frac{1}{M^\ell} \sum_{\theta_1, \ldots, \theta_\ell \in \Theta} \norm{C_{\theta_1} \ldots C_{\theta_\ell} S_B S_C B_{\theta_1} \ldots B_{\theta_\ell}} \quad\text{(by commutation)}\\
    &= \frac{1}{M^\ell} \sum_{\theta_1, \ldots, \theta_\ell \in \Theta} \norm{S_B S_C} \quad\text{(unitaries preserve spectral norm)}\\
    &\leq \epsilon\quad\text{(Lemma~\ref{lem:arbitrary-word},)}
\end{align*}
as claimed.
\end{proof}

\begin{corollary}[Bounding a ``disagreement-regularized'' $S_BS_C$]\label{cor:cross-resolvent}
For any $\delta > 0$, we have:
\begin{equation}\label{eq:cross-resolvent}
  \norm{S_B (D+\delta \id)^{-1} S_C}\leq \epsilon/\delta.
\end{equation}
\end{corollary}

\begin{proof}
First, note that since $D \succeq 0$ and $\delta > 0$, the inverse $(D+\delta \id)^{-1}$ is well-defined, Hermitian, and positive-definite.
The idea is to expand out $(D + \delta \id)^{-1}$ as a power series.
To this end:
\begin{align*}
    D + \delta  \id &= \frac{\id - (\id-2D)}{2} + \delta  \id \\
    &= \frac{1}{2} ((1+2\delta) \id - (\id-2D)) \\
    &= \frac{1+2\delta}{2} \left(\id - \frac{\id-2D}{1+2\delta}\right).
\end{align*}
Remembering that $\id-2D$ is an average of Hermitian reflections $B_\theta C_\theta$, it follows that $\norm{\id-2D} \leq 1 \Rightarrow \norm{(\id-2D)/(1+2\delta)} < 1$.
Therefore we can indeed compute the inverse with a power series as follows:
\begin{align*}
    (D+\delta \id)^{-1} &= \frac{2}{1+2\delta} \left(\id - \frac{\id-2D}{1+2\delta}\right)^{-1} \\
    &= \frac{2}{1+2\delta} \sum_{\ell = 0}^\infty \frac{(\id-2D)^\ell}{(1+2\delta)^\ell} \\
    \Rightarrow \norm{S_B (D+\delta \id)^{-1} S_C} &\leq \frac{2}{1+2\delta} \sum_{\ell = 0}^\infty \frac{\norm{S_B (\id-2D)^\ell S_C}}{(1+2\delta)^\ell} \\
    &\leq \frac{2}{1+2\delta} \sum_{\ell = 0}^\infty \frac{\epsilon}{(1+2\delta)^\ell} \text{ (Corollary~\ref{cor:Z-moments})} \\
    &= \frac{2\epsilon}{1+2\delta} \cdot \frac{1}{1 - \frac{1}{1+2\delta}} \\
    &= \epsilon/\delta.
\end{align*}
\end{proof}

\subsection{Completing the Proof}

The below proof makes use of Theorem~\ref{thm:separator}, an isolated linear-algebraic lemma that we defer to Section~\ref{sec:separator}.

\begin{proof}[Proof of Theorem~\ref{thm:security}]
Lemma~\ref{lem:game-decomposition} gives $D \succeq \pm\frac{1}{2}(S_B-S_C)$, and
Corollary~\ref{cor:cross-resolvent} with $\delta = \sqrt{\epsilon}$ gives
\[
  \norm{S_B(D+\delta \id)^{-1}S_C}\leq \delta.
\]
Applying Theorem~\ref{thm:separator} with $\gamma = \delta$ yields:
\[
  \frac{S_B+S_C}{2}\preceq D+\delta\id,
  \qquad\text{or equivalently}\qquad
  \frac{S_B+S_C}{2} - D\preceq \delta\id.
\]
The exact game decomposition in Lemma~\ref{lem:game-decomposition} therefore implies
\begin{equation}\label{eq:Xi-final}
  \Xi\preceq
  \left(\frac12+\frac{\delta}{2}\right)\id.
\end{equation}
Now applying Lemma~\ref{lem:winning-operator} and recalling that $J(\Phi)$ is a density operator tells us that:
\[
  \omega = \Tr{\Xi J(\Phi)}
  \leq\frac12+\frac{\delta}{2} = \frac12+\frac{1}{2} \sqrt{\frac{\norm{\Gamma}}{|\Theta|}},
\]
establishing the theorem.
\end{proof}

\begin{proof}[Proof of Main Theorem~\ref{thm:main-intro}]
  We instantiate the construction of Section~\ref{sec:construction} with the Pauli instantiation detailed therein.
  Correctness and efficiency have already been addressed.
  Finally, we invoke the security bound of Theorem~\ref{thm:security} and use Lemma~\ref{lem:symplectic-matrix} to deduce that $\norm{\Gamma^{\mathrm{Pauli}}} \leq 2^n$.
  Plugging this in implies the conclusion.
\end{proof}

\section{Key Linear-Algebraic Lemma}\label{sec:separator}

We isolate here the technical operator-theoretic workhorse (Theorem~\ref{thm:separator}) of the security proof; see Section~\ref{sec:separatoroverview} for some motivation and Section~\ref{sec:fullsecurity} for details on how this appears in the security proof.
Before that, we have a simple preliminary lemma. This is a well-known application of the Schur complement technique, but we include a proof for completeness.
\begin{lemma}[PSD $2 \times 2$ block matrices]\label{lem:2x2}
    Let $A, B, C$ be square Hermitian operators of the same dimension. Then:
    \begin{enumerate}
        \item\label{item:schurblocks} If $C$ is positive-definite, then $\begin{bmatrix} A & B \\ B & C \end{bmatrix}$ is PSD if and only if $A - BC^{-1}B$ is PSD.
        \item\label{item:hadamardblocks} The block matrix $\begin{bmatrix} A & B \\ B & A \end{bmatrix}$ is PSD if and only if $A-B$ and $A+B$ are both PSD.
    \end{enumerate}
\end{lemma}
\begin{proof}
    To prove Item~\ref{item:schurblocks}, 
    note that if we let $T = \begin{bmatrix} \id & 0 \\ C^{-1}B & \id \end{bmatrix}$, we have:
    \begin{align*}
        T^* \begin{bmatrix} A-BC^{-1}B & 0 \\ 0 & C \end{bmatrix} T &= \begin{bmatrix} \id & BC^{-1} \\ 0 & \id \end{bmatrix} \begin{bmatrix} A-BC^{-1}B & 0 \\ 0 & C \end{bmatrix}  \begin{bmatrix} \id & 0 \\ C^{-1}B & \id \end{bmatrix}\\
        &= \begin{bmatrix} \id & BC^{-1} \\ 0 & \id \end{bmatrix} \begin{bmatrix} A-BC^{-1}B & 0 \\ B & C \end{bmatrix} \\
        &= \begin{bmatrix} A & B \\ B & C \end{bmatrix}.
    \end{align*}
    Since $T$ is invertible, it follows that $\begin{bmatrix} A & B \\ B & C \end{bmatrix}$ being PSD is equivalent to $\begin{bmatrix} A-BC^{-1}B & 0 \\ 0 & C \end{bmatrix}$ being PSD, which is in turn equivalent to $A-BC^{-1}B$ and $C$ both being PSD.
    We are already assuming that $C$ is positive-definite, so this is in turn equivalent to $A-BC^{-1}B$ being PSD.
    
    Now we turn to Item~\ref{item:hadamardblocks}. Note that $U = \frac{1}{\sqrt{2}} \begin{bmatrix} \id & \id \\ \id & -\id \end{bmatrix}$ is invertible.
    Conjugating by $U$, we have:
    \begin{align*}
        U \begin{bmatrix} A & B \\ B & A \end{bmatrix} U^* &= \frac{1}{2} \begin{bmatrix} \id & \id \\ \id & -\id \end{bmatrix} \begin{bmatrix} A & B \\ B & A \end{bmatrix} \begin{bmatrix} \id & \id \\ \id & -\id \end{bmatrix} \\
        &= \frac{1}{2} \begin{bmatrix} \id & \id \\ \id & -\id \end{bmatrix} \begin{bmatrix} A+B & A-B \\ A+B & B-A \end{bmatrix} \\
        &= \begin{bmatrix} A+B & 0 \\ 0 & A-B \end{bmatrix}.
    \end{align*}
    So $\begin{bmatrix} A & B \\ B & A \end{bmatrix}$ is PSD if and only if $\begin{bmatrix} A+B & 0 \\ 0 & A-B \end{bmatrix}$ is PSD, which is in turn true if and only if $A+B$ and $A-B$ are both PSD.
\end{proof}

\begin{theorem}[Cloning advantage is bounded by ``disagreement-regularized'' $S_BS_C$]\label{thm:separator}
Let $S_B, S_C, D$ be self-adjoint operators with $D \succeq \frac{S_B-S_C}{2}, \frac{S_C-S_B}{2}$.
For reals $\delta, \gamma > 0$, assume the inequality:
\begin{equation}\label{eq:separator-hyp}
  \norm{S_B(D+\delta \id)^{-1}S_C}\leq\gamma.
\end{equation}
Then we have:
\begin{equation}\label{eq:separator-conclusion}
  \frac{S_B+S_C}{2} - D \preceq \max(\delta, \gamma) \cdot \id.
\end{equation}
\end{theorem}

\begin{proof}
    The two positivity assumptions imply $D\succeq0$, so $(D+\delta \id)^{-1}$ is well-defined and positive definite.
    Let
    \[
      K=S_B(D+\delta \id)^{-1}S_C.
    \]
    We know that $\norm{K} \leq \gamma$ and therefore $\norm{K^*} \leq \gamma$ also.
    It follows that we have:
    \begin{align*}
        \gamma &\geq \frac{1}{2} \norm{K+K^*} \\
        &= \frac{1}{2} \norm{S_B(D+\delta \id)^{-1}S_C + S_C(D+\delta \id)^{-1}S_B} \\
        &= \frac{1}{4} \norm{(S_B+S_C)(D+\delta \id)^{-1}(S_B+S_C) - (S_B-S_C)(D+\delta \id)^{-1}(S_B-S_C)},
    \end{align*}
    implying:
    \begin{equation}\label{eq:YFY-TFT}
      (S_B+S_C)(D+\delta \id)^{-1}(S_B+S_C)\preceq (S_B-S_C)(D+\delta \id)^{-1}(S_B-S_C)+4\gamma\id.
    \end{equation}
    Next, Item~\ref{item:hadamardblocks} of Lemma~\ref{lem:2x2} tells us that $$\begin{bmatrix} D & \frac{1}{2}(S_B-S_C) \\ \frac{1}{2}(S_B-S_C) & D \end{bmatrix}$$
    is PSD.
    Consequently, we also have:
    $$\begin{bmatrix} D & \frac{1}{2}(S_B-S_C) \\ \frac{1}{2}(S_B-S_C) & D+\delta \id \end{bmatrix} = \begin{bmatrix} D & \frac{1}{2}(S_B-S_C) \\ \frac{1}{2}(S_B-S_C) & D \end{bmatrix} + \begin{bmatrix} 0 & 0 \\ 0 & \delta \id \end{bmatrix} \succeq 0.$$
    Applying Item~\ref{item:schurblocks} of Lemma~\ref{lem:2x2} to this matrix implies that:
    \begin{align*}
        D-\frac{1}{4}(S_B-S_C)(D+\delta \id)^{-1}(S_B-S_C) &\succeq 0 \\ \Rightarrow D+\max(\delta, \gamma) \cdot \id &\succeq \frac{1}{4}(S_B-S_C)(D+\delta \id)^{-1}(S_B-S_C) + \max(\delta, \gamma) \cdot \id \\
        &\succeq \frac{1}{4}(S_B+S_C)(D+\delta \id)^{-1} (S_B+S_C) \text{ (plugging in~\eqref{eq:YFY-TFT})} \\
        &\succeq \frac{1}{4} (S_B+S_C)(D + \max(\delta, \gamma) \cdot \id)^{-1} (S_B+S_C).
    \end{align*}
    The last step above uses the fact that $D+\delta \id \preceq D+\max(\delta, \gamma) \cdot \id$ and both these operators are positive definite, so $(D+\max(\delta, \gamma) \cdot \id)^{-1} \preceq (D+\delta \id)^{-1}$.
    Now, Item~\ref{item:schurblocks} of Lemma~\ref{lem:2x2} tells us that
    $$\begin{bmatrix} D+\max(\delta, \gamma) \cdot \id & \frac{1}{2}(S_B+S_C) \\ \frac{1}{2}(S_B+S_C) & D+\max(\delta, \gamma) \cdot \id \end{bmatrix}$$
    is PSD.
    Finally, Item~\ref{item:hadamardblocks} of Lemma~\ref{lem:2x2} now tells us that $D+\max(\delta, \gamma) \cdot \id \succeq \frac{1}{2}(S_B+S_C)$, as desired.
\end{proof}

\section{An Alternate Proof for the Haar Scheme}\label{sec:haar}

Here, we apply the techniques by~\cite{MajenzSchaffnerTahmasbi2021,broadbent2025optimaluntelegraphableencryptionimplications,PorembaRagavanVaikuntanathan2026} to show that the Haar scheme proposed by~\cite{BhattacharyyaBroadbentCulf2026}, though inefficient, also achieves security $1/2 + O(2^{-n/2})$.
We prove this here in the interest of self-containment; we do not make any claim to originality in this section.
The original proof by~\cite{BhattacharyyaBroadbentCulf2026} showed security $1/2+O(2^{-n/8})$.

\begin{construction}\label{con:haar}
Fix $n \geq 1$ and put $d = 2^n$.
Let $\mathfrak{U}$ be any compact subgroup of the $d$-dimensional unitary group.
\begin{itemize}
  \item Key generation: sample $U$ from the Haar measure on $\mathfrak{U}$.
  \item To encrypt $a \in \{0, 1\}$, output
    \begin{equation}
        \rho_{U, a} = U\left(\proj{a} \otimes \frac{2}{d} \id_{d/2}\right)U^*.
    \end{equation}
    \item To decrypt, apply $U^*$ and measure the first qubit in the standard basis.
\end{itemize}
\end{construction}

To compare this with our scheme, note that our Pauli-based scheme in Section~\ref{sec:construction} can alternately be formulated in terms closer to Construction~\ref{con:haar}.

\begin{lemma}
    Assume Construction~\ref{con:pauli} is instantiated with the $4^n-1$ non-identity Paulis.
    Then for every key $\theta \in \Theta_n$, there exists an $n$-qubit Clifford $U_\theta$ such that the encryption/decryption operations in Construction~\ref{con:haar} for $U_\theta$ are identical to the encryption/decryption operations in Construction~\ref{con:pauli} for the key $\theta$.
\end{lemma}
\begin{proof}
    The observation is that for any unitary $U$, we have:
    \begin{align*}
        U\left(\proj{a} \otimes \frac{2}{d} \id_{d/2}\right)U^* &= U\left(\frac{\id + (-1)^a Z}{2} \otimes \frac{2}{d} \id_{d/2}\right)U^* \\
        &= U\left(\frac{1}{d} \id_d + \frac{1}{d} (-1)^a (Z \otimes \id_{d/2})\right)U^* \\
        &= \frac{1}{d} \left(\id_d + (-1)^a U(Z \otimes \id_{d/2})U^*\right).
    \end{align*}
    Since $n$-qubit Clifford group acts transitively on the non-identity Hermitian Paulis, there exists an $n$-qubit Clifford $U_\theta$ such that $U_\theta(Z \otimes \id_{d/2})U_\theta^* = P_\theta$.
    The conclusion follows.
\end{proof}

\begin{theorem}\label{thm:haarmain}
    Let $\mathfrak{U}$ be any compact subgroup of the $d$-dimensional unitary group that contains the $n$-qubit Clifford group.
    Then Construction~\ref{con:haar} instantiated with $\mathfrak{U}$ is also secure; any adversaries achieve joint cloning success probability at most
    $$\frac{1}{2} + \frac{1}{2} \sqrt{\frac{2^n}{4^n-1}}.$$
    (In particular, this holds when $\mathfrak{U}$ is the entire unitary group.)
\end{theorem}
\begin{proof}
    Consider any adversary $(\mathsf{Cl}, \mathsf{B}, \mathsf{C})$ for Construction~\ref{con:haar} instantiated with $\mathfrak{U}$, and assume that Bob and Charlie simultaneously guess correctly with probability $\omega$.
    From these, we build an adversary $(\mathsf{Cl}', \mathsf{B}', \mathsf{C}')$, for Construction~\ref{con:pauli} instantiated with the Paulis.
    $(\mathsf{Cl}', \mathsf{B}', \mathsf{C}')$ have some unitary $U$ hardcoded, that we will address at the end.
    \begin{itemize}
        \item $\mathsf{Cl}'$: when they receive the state $\rho$, run $\mathsf{Cl}$ on the state $U\rho U^*$ and forward the resulting bipartite state to $\mathsf{B}'$ and $\mathsf{C}'$.

        \item $\mathsf{B}'$: once the Pauli key $\theta$ is revealed, reveal $UU_\theta$ as the key to $\mathsf{B}$ and output the guess produced by $\mathsf{B}$.

        \item $\mathsf{C}'$: analogously for $\mathsf{C}$.
    \end{itemize}
    First, note that the view of $(\mathsf{Cl}, \mathsf{B}, \mathsf{C})$ is as if Construction~\ref{con:haar} has been instantiated with the key $UU_\theta$; this is clear for Bob and Charlie, and the state received by $\mathsf{Cl}$ is:
    $$U\rho U^* = U\left(U_\theta \left(\proj{a} \otimes \frac{2}{d} \id_{d/2}\right)U_\theta^*\right)U^* = UU_\theta \left(\proj{a} \otimes \frac{2}{d} \id_{d/2}\right)(UU_\theta)^*.$$
    Bearing this in mind, we have the following chain of implications:
    \begin{itemize}
        \item Main Theorem~\ref{thm:main-intro} tells us that $(\mathsf{Cl}', \mathsf{B}', \mathsf{C}')$ with $U$ hardcoded succeed with probability at most $\frac{1}{2} + \frac{1}{2} \sqrt{2^n/(4^n-1)}$.
        \item Equivalently, for every fixed $U$, $(\mathsf{Cl}, \mathsf{B}, \mathsf{C})$ succeed on the key ensemble $\{UU_\theta: \theta \gets \Theta_n\}$ with probability at most $\frac{1}{2} + \frac{1}{2} \sqrt{2^n/(4^n-1)}$.
        \item Averaging over $U$ now tells us that $(\mathsf{Cl}, \mathsf{B}, \mathsf{C})$ succeed on the key ensemble $\{UU_\theta: \theta \gets \Theta_n, U \sim \mathrm{Haar}(\mathfrak{U})\}$ with probability at most $\frac{1}{2} + \frac{1}{2} \sqrt{2^n/(4^n-1)}$.
        \item Finally, Haar invariance tells us that this is exactly equivalent to $\omega \leq \frac{1}{2} + \frac{1}{2} \sqrt{2^n/(4^n-1)}$, as desired.
    \end{itemize}
\end{proof}

\bibliographystyle{alpha}
\bibliography{main}

\end{document}